\documentclass[11pt]{article}
\usepackage{verbatim}
\usepackage{fullpage}

\usepackage{amssymb}
\usepackage{framed}
\usepackage{epsfig}
\usepackage{amsmath, amsthm, amssymb}

\usepackage{verbatim}

\newtheorem{theorem}{Theorem}[section]

\newtheorem{corollary}[theorem]{Corollary}

\newtheorem{lemma}[theorem]{Lemma}

{\mbox{}\\[5pt]}

\newcommand{\ignore}[1]{}
\newcommand{\eat}[1]{}

\title{Complexity of Scarf's Lemma and Related Problems}

\vspace{0.15in}

\author{{Shiva Kintali}\footnote{College of Computing,
Georgia Institute of Technology, Atlanta, GA-30332. Email :
{\em{kintali@cc.gatech.edu}}}
}
\date{}
\begin{document}
\maketitle

\begin{abstract}

Scarf's lemma is one of the fundamental results in combinatorics, originally introduced to study the core of an $N$-person game. Over the last four decades, the usefulness of Scarf's lemma has been demonstrated in several important combinatorial problems seeking {\em{stable}} solutions (\cite{fbgp}, \cite{onlemma} \cite{frackernel}, \cite{h-perfect}). However, the complexity of the computational version of Scarf's lemma ({\sc{Scarf}}) remained open. In this paper, we prove that {\sc{Scarf}} is complete for the complexity class {\bf{PPAD}}. This proves that {\sc{Scarf}} is as hard as the computational versions of Brouwer's fixed point theorem and Sperner's lemma. Hence, there is no polynomial-time algorithm for {\sc{Scarf}} unless {\bf{PPAD}} $\subseteq$ {\bf{P}}. We also show that fractional stable paths problem and finding strong fractional kernels in digraphs are {\bf{PPAD}}-hard. \\

{\bf{Keywords}}: core in a balanced game, fractional stable paths problem, fractional kernels in digraphs, hypergraphic preference systems, Nash-equilibrium, PPAD-completeness, Scarf's lemma.

\end{abstract}

\newpage

\section{Introduction}

The study of combinatorial problems seeking {\em{stable}} solutions has a long history, dating back to stable marriage problem. In 1962, Gale and Shapley presented a polynomial-time algorithm for finding a stable matching in a bipartite graph \cite{gale-shapley}. Since then, several generalizations and extensions of stable marriage problem have been studied. In this paper, we study the complexity of several combinatorially rich problems having similar flavor. These problems include {\em{stable paths problem}} \cite{fbgp}, finding stable matchings in {\em{hypergraphic preference systems}} \cite{onlemma} and computing {\em{kernels in directed graphs}} (\cite{frackernel}, \cite{h-perfect}). Although they are defined and studied in different contexts, they have one strikingly common feature. While there are instances of these problems with no stable solution, every instance is {\em{guaranteed}} to have a {\em{fractional}} stable solution. The existence of these fractional stable solutions is proved using Scarf's lemma, originally introduced to study the core of an $N$-person game \cite{scarf}.

\subsection{Scarf's Lemma}

In one of the fundamental papers in game theory \cite{scarf}, Scarf studied the core of an $N$-person game and proved that {\em{every balanced $N$-person game with nontransferable utilities has a nonempty core}}. His proof is based on an elegant combinatorial argument, and makes use of no fixed point theorems. The core (no pun intended) of his argument, known as Scarf's lemma, found many independent applications in a diverse set of combinatorial problems (\cite{fbgp}, \cite{onlemma} \cite{frackernel}, \cite{h-perfect}). Below we state Scarf's lemma \cite{scarf} and define the computational version of Scarf's lemma ({\sc{Scarf}}). Let $I = [\delta_{ij}]$ be an $m \times m$ identity matrix. Let $[n] = \{1,2,\dots,n\}$.

\begin{theorem}\label{thm:lem}\mbox{(Scarf's lemma)}
Let $m < n$ and let $B$ be an $m \times n$ real matrix such that $b_{ij} = \delta_{ij}$ for $1 \leqslant i, j \leqslant m$. Let $b$ be a non-negative vector in ${\mathbb{R}}^m$, such that the set $\{\alpha\in{\mathbb{R}}_{+}^n : B{\alpha}=b\}$ is bounded. Let $C$ be an $m \times n$ matrix such that $c_{ii} \leqslant c_{ik} \leqslant c_{ij}$ whenever $i,j \leqslant m$, $i \neq j$ and $k > m$. Then there exists a subset $J$ of size $m$ of $[n]$ such that
\begin{itemize}
\item{$B{\alpha}=b$ for some $\alpha\in{\mathbb{R}}_{+}^n$ such that $\alpha_j=0$ whenever $j\notin{J}$, and}
\item{For every $k \in [n]$ there exists $i \in [m]$ such that $c_{ik} \leqslant c_{ij}$ for all $j \in J$.}
\end{itemize}
\end{theorem}

\begin{framed}
\noindent {\sc{Scarf}} : Given matrices $B$, $C$ and a vector $b$ satisfying the conditions in the above theorem, find $\alpha\in{\mathbb{R}}_{+}^n$ satisfying the conditions above.
\end{framed}

\subsection{PPAD}

A search problem $\mathcal{S}$ is a set of inputs $I_{\mathcal{S}} \subseteq \Sigma^*$ such that for each $x \in I_{\mathcal{S}}$ there is an associated set of solutions ${\mathcal{S}}_x \subseteq \Sigma^{{|x|}^k}$ for some integer $k$, such that for each $x \in I_{\mathcal{S}}$ and $y \in \Sigma^{{|x|}^k}$ whether $y \in {\mathcal{S}}_x$ is decidable in polynomial time. A search problem is {\em{total}} if ${\mathcal{S}}_x \neq \emptyset$ for all $x \in I_{\mathcal{S}}$. {\bf{TFNP}} is the set of all total search problems \cite{class-tfnp}. Since TFNP is a semantic class, several syntactic classes (e,g., PLS \cite{class-pls}, PPA, PPAD, PPP, PPM \cite{class-ppad}) were defined to study the computational phenomenon of TFNP. Since every member of TFNP is equipped with a mathematical proof that it belongs to TFNP, these syntactic classes are defined based on their {\em{proof styles}}. The complexity class {\bf{PPAD}}, introduced by Papadimitriou \cite{class-ppad}, is the class of all search problems whose totality is proved using a {\em{parity argument}}. These search problems are reducible to the following problem :

\begin{framed}
\noindent {\sc{End Of The Line}} : Given two boolean circuits $S$ and $P$ with $n$ input bits and $n$ output bits, such that $P(0^n) = 0^n \neq S(0^n)$, find an input $x \in \{0,1\}^n$ such that $P(S(x)) \neq x$ or $S(P(x)) \neq x \neq 0^n$.
\end{framed}

A polynomially computable function $f$ is a polynomial-time reduction from total search problem
$\mathcal{S}$ to total search problem $\mathcal{T}$ if for every input $x$ of $\mathcal{S}$, $f(x)$ is an input of $\mathcal{T}$, and furthermore there is another polynomial function $g$ such that for every $y \in \mathcal{T}_{f(x)}, g(y) \in \mathcal{S}_x$. A search problem $\mathcal{S}$ in PPAD is called {\bf{PPAD}}-{\em{complete}} if all problems in PPAD reduce to it in polynomial-time.

\subsection{Related Work and Our Contributions}

Aharoni and Holzman \cite{frackernel} proved that every clique-acyclic digraph has a strong fractional kernel. Aharoni and Fleiner \cite{onlemma} proved that every hypergraphic preference system has a fractional stable matching. Both these proofs are based on Scarf's lemma. Haxell and Wilfong \cite{fbgp} proved that every instance of fractional stable paths problem (FSPP) has a fractional stable solution. Their proof works in two stages. In the first stage, they use Scarf's lemma to show that every instance of FSPP has an $\epsilon$-solution, for any positive constant $\epsilon$. Then they apply a standard compactness-type argument to conclude that every instance has an exact solution.

The complexity class TFNP was introduced by Megiddo and Papadimitriou \cite{class-tfnp}. The class PLS (for polynomial local search) was introduced by Johnson, Papadimitriou and Yannakakis \cite{class-pls}. In \cite{class-ppad} Papadimitriou defined the complexity classes PPA (polynomial parity argument), PPAD (polynomial parity argument in directed graphs), PPP (polynomial pigeon-hole principle) and PPM (polynomial probabilistic method). He proved that computational versions of Brouwer's fixed point theorem and Sperner's lemma are PPAD-complete.

Daskalakis, Goldberg and Papadimitriou \cite{nash-ppad} proved that {\sc{3-Dimensional Brouwer}} is PPAD-complete. They reduced {\sc{3-dimensional Brouwer}} to {\sc{3-Graphical Nash}} to prove that {\sc{3-Graphical Nash}} is PPAD-complete. These results together with the reductions of Goldberg and Papadimitriou \cite{normal-graphic} imply that computing Nash equilibrium in games with 4 players ({\sc{4-Nash}}) is PPAD-complete. Chen and Deng \cite{3nash-chen}, Daskalakis and Papadimitriou \cite{3nash-papa} independently proved that {\sc{3-Nash}} is PPAD-complete. Chen and Deng \cite{2nash-chen} proved that {\sc{2-Nash}} is also PPAD-complete. Chen, Deng and Teng \cite{apxnash-chen} showed that {\em{approximating}} Nash equilibrium is also hard. For a list results on the complexity of finding equilibria we refer the reader to a recent book on algorithmic game theory \cite{agt-book}.

In this paper, we study the complexity of the computational version of Scarf's lemma ({\sc{Scarf}}). We prove that {\sc{Scarf}} is complete for the complexity class PPAD. This proves that {\sc{Scarf}} is as hard as the computational versions of Brouwer's fixed point theorem and Sperner's lemma. Hence, there is no polynomial-time algorithm for {\sc{Scarf}} unless PPAD $\subseteq$ P. We also show that fractional stable paths problem and finding strong fractional kernels in digraphs are PPAD-hard. In Section \ref{sec:related}, we mention several related problems belonging to the complexity class PPAD.

\section{Kernels in Digraphs}\label{sec:kernel}

The problem of finding kernels in digraphs plays a crucial role in all our proofs in this paper. In this section, we define computational problems related to finding kernels in digraphs. Let $D(V, A)$ be a directed graph. Let $I(v)$ denote the in-neighborhood of a vertex $v$ i.e., $I(v)$ is $v$ together with the vertices $u$ such that $(u, v) \in A$. A set $K$ of vertices is a clique if every two vertices in K are connected by at least one arc. A set of vertices is called {\em{independent}} if no two distinct vertices in it are connected by an arc. A subset of $V$ is called {\em{dominating}} if it meets $I(v)$ for every $v \in V$. A {\em{kernel}} in $D$ is an independent and dominating set of vertices. A directed triangle shows that not all digraphs have kernels.

A non-negative function $f$ on $V$ is called {\em{fractionally dominating}} if $\sum_{u \in I(v)}{f(u)} \geqslant 1$ for every vertex $v$. The function is {\em{strongly dominating}} if for all $v$, $\sum_{u \in K}{f(u)} \geqslant 1$ for some clique $K$ contained in $I(v)$. A non-negative function $f$ on $V$ is called {\em{fractionally independent}} if $\sum_{u \in K}{f(u)} \leqslant 1$ for every clique $K$. A {\em{fractional kernel}} is a function on $V$ which is {\em{both}} fractionally independent and fractionally dominating. In case that it is also strongly dominating, it is called a {\em{strong}} fractional kernel. A directed triangle shows that not all digraphs have fractional kernels.

An arc $(u, v)$ is called {\em{irreversible}} if $(v, u)$ is not an arc of the graph. A cycle in $D$ is called {\em{proper}} if all of its arcs are irreversible.  A digraph in which no clique contains a proper cycle is called {\em{clique-acyclic}}. The following theorem was proved by Aharoni and Holzman \cite{frackernel} using Scarf's lemma.

\begin{theorem}\label{clique-acyclic} (Aharoni and Holzman \cite{frackernel})
Every clique-acyclic digraph has a strong fractional kernel.
\end{theorem}

\begin{framed}
\noindent {\sc{Strong Kernel}} : Given a clique-acyclic digraph $D(V,E)$, find a strong fractional kernel.
\end{framed}

\subsection{A Game-theoretic Kernel Problem}

We define two variants of {\sc{Strong Kernel}} that play a crucial role in our reductions in the next sections. Two cycles in $D$ are said to be node-disjoint if they do not have any nodes in common. A proper cycle $C$ in a digraph $D(V,A)$ is called {\em{homogeneous}} if for every $v \in C$, $u \notin C$, $(u,v) \in A$ implies $(u,v') \in A$ for all $v' \in C$.

\begin{framed}
\noindent {\sc{3-Strong Kernel}} : Given a clique-acyclic digraph $D(V,E)$, in which every maximal clique is of size at most 3, and all proper cycles are homogeneous and node-disjoint, find a strong fractional kernel.
\end{framed}

Let $F = \{f(u)\ |\ u \in V\}$ be a fractional kernel of a digraph $D(V,A)$. We look at $F$ from a game-theoretic perspective. The nodes of the digraph represent {\em{players}} and $f(u)$ represents the {\em{cost}} incurred by player $u$. A function $F$ is said to be a Nash equilibrium if no player (say $u$) can decrease its cost ($f(u)$) {\em{unilaterally}} without violating the conditions (involving $u$) of fractional kernel.

\begin{framed}
\noindent {\sc{3-Kernel Nash}} : Given a clique-acyclic digraph $D(V,E)$, in which every maximal clique is of size at most 3, and all proper cycles are homogeneous and node-disjoint, find a fractional kernel that is a Nash equilibrium.
\end{framed}

\subsection{Our Approach}

We write {\sc{A}} $\leq_P$ {\sc{B}} to say that {\sc{A}} is polynomial-time reducible to {\sc{B}}. In Section \ref{sec:ppad}, we prove that {\sc{3-Kernel Nash}} $\leq_P$ {\sc{3-Strong Kernel}} $\leq_P$ {\sc{Scarf}} $\leq_P$ {\sc{End Of The Line}}. This proves that {\sc{3-Kernel Nash}}, {\sc{3-Strong Kernel}} and {\sc{Scarf}} are in PPAD. In Section \ref{sec:ppadcomplete}, we prove that {\sc{3-Dimensional Brouwer}} $\leq_P$ {\sc{3-Kernel Nash}}. Since {\sc{3-Dimensional Brouwer}} is PPAD-complete \cite{nash-ppad}, this implies that {\sc{3-Kernel Nash}}, {\sc{3-Strong Kernel}} and {\sc{Scarf}} are PPAD-complete. In Section \ref{sec:fspp}, we prove that fractional stable paths problem ({\sc{Fspp}}) is PPAD-hard. We do this by showing that {\sc{3-Kernel Nash}} $\leq_P$ {\sc{Fspp}}.

\section{Reductions to {\sc{End Of The Line}}}\label{sec:ppad}

In this Section, we prove that {\sc{3-Kernel Nash}} $\leq_P$ {\sc{3-Strong Kernel}} $\leq_P$ {\sc{Scarf}} $\leq_P$ {\sc{End Of The Line}}.

\begin{lemma}\label{lem:threegame}
{\sc{3-Kernel Nash}} $\leq_P$ {\sc{3-Strong Kernel}}.
\end{lemma}
\begin{proof}
The input digraphs to {\sc{3-Kernel Nash}} and {\sc{3-Strong Kernel}} are the same (say $D$). Let $W = \{w(u)\ |\ u \in V\}$ be a solution to {\sc{3-Strong Kernel}}. Given $W$, the algorithm ${\bf{Compute Nash}}(W)$ finds a solution to {\sc{3-Kernel Nash}} in polynomial-time.

\noindent \line(1,0){400} \\
\noindent ${\bf{Compute Nash}}(W)$ \\
\indent Find a maximum weight cycle cover of $D$, (the weight of each edge is set to unity). \\
\indent Identify each homogeneous cycle into {\em{super node}} and remove multiple edges. \\
\indent While there is a node $v$ such that $\sum_{u \in I(v)}{w(u)} = 1 + \delta$, $\delta > 0$ and $w(v) > 0$\\
\indent \{ \\
\indent \indent $w(v) := w(v) - min(\delta, w(v))$ \\
\indent \indent For each $v' \neq v$ such that $v \in I(v')$ \\
\indent \indent \indent if $\sum_{u \in I(v')}{w(u)} = 1 - \alpha$, $\alpha > 0$ \\
\indent \indent \indent \indent $w(v') := w(v') + \alpha$ \\
\indent \} \\
\indent Restore the multiple edges. \\
\indent Expand the super nodes into original nodes of $D$. \\
\indent For each node $u$ in a super node $v$ \\
\indent \indent $w(u) := w(v)/2$ \\
\indent return $W$ \\
\noindent \line(1,0){400} \\

A cycle cover of a graph is a set of cycles such that every vertex is part of exactly one cycle.
Weight of a cycle cover is the sum of weights of the {\em{edges}} of the cycles. Finding maximum weight cycle cover of $D$ takes polynomial time. Since there are no cycles (after identifying homogeneous cycles) each iteration of the while loop takes at most $O(|V|)$ time.
\end{proof}

\begin{lemma}\label{lem:kernelppad}
{\sc{3-Strong Kernel}} $\leq_P$ {\sc{Scarf}}.
\end{lemma}
\begin{proof}
Aharoni and Holzman \cite{frackernel} proved Theorem \ref{clique-acyclic} by constructing matrices $B$ and $C$ from the digraph $D$ and a vector $b$ of all 1's and appealing to Scarf's lemma. The rows of $B$ and $C$ are indexed by the set of maximal cliques in $D$. In an instance of {\sc{3-Strong Kernel}} all maximal cliques are of size at most 3. Hence, the number of rows is polynomially bounded in the size of $D$. Hence {\sc{3-Strong Kernel}} is polynomial-time reducible to {\sc{Scarf}}.
\end{proof}

Now we show that {\sc{Scarf}} reduces to the {\sc{End Of The Line}}. The reduction is essentially the original proof of Scarf's lemma \cite{scarf} as appeared in \cite{frackernel}. Let $J \subseteq [n]$. A column $c_k$ of $C$ is said to be {\em{$J$-subordinated at the index $i$}} if $c_{ik} \leqslant c_{ij}$ for every $j \in J$. It is said to be {\em{$J$-subordinated}} if it is $J$-subordinated at some $i$. We say that $J$ is {\em{subordinating}} if every column of $C$ is $J$-subordinated. Note that if $J' \subseteq J$ and $J$ is subordinating for $C$ then so is $J'$. A subset $J$ of size $m$ of [$n$] is called a feasible basis of ($B, b$) if the columns $b_j, j \in J$, are linearly independent, and there exist non-negative numbers $\alpha_j, j \in J$, such that $\sum_{j \in J}{{\alpha_j}{b_j}} = b$. In other words, $b$ belongs to the cone spanned by the columns $b_j, j \in J$.

The pair ($B, b$) is {\em{non-degenerate}} if $b$ is not in the cone spanned by fewer than $m$ columns of $B$. We call $C$ ordinal-generic if all the elements in each row of $C$ are distinct. There exists a small perturbation $b'$ of $b$ such that the pair ($B, b$) is non-degenerate and every feasible basis for ($B, b$) is also a feasible basis for ($B, b$). By slightly perturbing $C$ we can obtain an ordinal-generic matrix $C'$ satisfying
the assumptions of the theorem, and if the perturbation is small enough then any subordinating set for $C'$ is also subordinating for $C$. Hence, we may assume that ($B, b$) is non-degenerate, and that $C$ is ordinal-generic.

Lemma \ref{lem:simplex} is well-known and is at the heart of the simplex algorithm.
Its proof requires that $\{\alpha\in{\mathbb{R}}_{+}^n : B{\alpha}=b\}$ is bounded and $(B,b)$ is non-degenerate. For the proof of Lemma \ref{lem:replace}, we refer the reader to \cite{scarf} or \cite{frackernel} or page 1127 of the three volume series of Schrijver's book \cite{schrijver-book}.

\begin{lemma}\label{lem:simplex}
Let $J$ be a feasible basis for $(B, b)$, and $k\in[n]\setminus{J}$. Then there exists a unique $j \in J$ such that $J+k-j$ $(i.e., J\cup\{k\}{\setminus}\{j\})$ is a feasible basis. Also, given $J$ and $k$, we can find $j$ in polynomial-time.
\end{lemma}

\begin{lemma}\label{lem:replace}(\cite{scarf})
Let $K$ be a subordinating set for $C$ of size $m$-1. Then there are precisely two elements $j \in [n]{\setminus}K$ such that $K+j$ is subordinating for $C$, unless $K\subseteq[m]$, in which case there exists precisely one such $j$.
\end{lemma}

\begin{theorem}\label{thm:scarf}
{\sc{Scarf}} $\leq_P$ {\sc{End Of The line}}.
\end{theorem}
\begin{proof}
We shall construct a bipartite graph $\mathcal{G}$ with bipartition $\mathcal{F}$ and $\mathcal{S}$, where $\mathcal{F}$ is the set of all feasible bases containing 1, and $\mathcal{S}$ is the set of all subordinating sets of size $m$ not containing 1. An element $F$ of $\mathcal{F}$ and an element $S$ of $\mathcal{S}$ are joined by an edge from $F$ to $S$ if $F{\setminus}S=\{1\}$.

We shall prove, using {\em{end of the line argument}}, that there exists a set $J$ of size $m$ which is both subordinating and a feasible basis. Consider a set $F \in \mathcal{F}$ which is not subordinating, and assume that $F$ has positive degree in $\mathcal{G}$. Then the set $K=F{\setminus}\{1\}$ is subordinating. Applying Lemma \ref{lem:replace} to $K$, we see that $F$ has degree 2 in $\mathcal{G}$, unless $F=[m]$, in which case it has degree 1. By the properties of the matrix $C$ mentioned in Theorem \ref{thm:lem} it is easy to see that $[m]$ is in $F$ and is not subordinating.

Now consider a set $S \in \mathcal{S}$ which is not a feasible basis, and assume that
$S$ has positive degree in $\mathcal{G}$. Let $F$ be a neighbor of $S$, and let $s$ be the single element of $S{\setminus}F$. By Lemma \ref{lem:simplex} there exists a unique element $f$ of $F$ such that $F'=F+s-f$ is a feasible basis. If $f=1$ then $F'=S$, which contradicts our assumption about $S$. Thus $f\neq{1}$, and $F'$ is the unique element of $\mathcal{F}$, different from $F$, which is connected to $S$.

Hence, every vertex of $\mathcal{G}$ which is not both subordinating and a feasible basis has degree 0 or 2, apart from $[m]$, which has outdegree 1. Similarly a vertex which is both subordinating and a feasible basis, if it exists, has degree 1. Thus the connected component of $\mathcal{G}$ containing $[m]$ is a path, which must {\em{end}} at another vertex of degree 1, i.e., at a vertex which is both subordinating and a feasible basis.

The above proof of Scarf's lemma uses an ``{\em{undirected}} end of the line argument", thus showing the PPA-membership of {\sc{Scarf}}. To prove PPAD-membership of Scarf, we need a ``{\em{directed}} end of the line
argument". Shapley \cite{ShapleyOrientation} presented a geometric orientation rule for the equilibrium points of (nondegenerate) bimatrix games based on the Lemke-Howson algorithm. Extending Shapley's rule, Todd \cite{ToddOrientation} developed a similar orientation theory for generalized complementary pivot algorithms. Applying Todd's orientation technique we can prove PPAD-membership of Scarf. For more details we refer the reader to \cite{kprst-ppad}.

\end{proof}

\begin{corollary}
{\sc{3-Kernel Nash}}, {\sc{3-Strong Kernel}} and {\sc{Scarf}} are in PPAD.
\end{corollary}
\begin{proof}
We have {\sc{3-Kernel Nash}} $\leq_P$ {\sc{3-Strong Kernel}} $\leq_P$ {\sc{Scarf}} $\leq_P$ {\sc{End Of The Line}}. Hence all these problems are in PPAD.
\end{proof}

\section{Proof of PPAD-completeness}\label{sec:ppadcomplete}

In this section we prove that {\sc{3-Kernel Nash}} is PPAD-complete. We show that {\sc{3-Dimensional Brouwer}} $\leq_P$ {\sc{3-Kernel Nash}}. Since {\sc{3-Dimensional Brouwer}} is PPAD-complete \cite{nash-ppad}, {\sc{3-Kernel Nash}} is also PPAD-complete.

We now present an outline of {\sc{3-Dimensional Brouwer}}. For more details we refer the reader to \cite{nash-ppad}. We are given a Brouwer function $\phi$ on the 3-dimensional unit cube, defined in terms of its values at the centers of $2^{3n}$ cubelets with side $2^{-n}$. At the center $c_{ijk}$ of the cubelet $K_{ijk}$ defined as
\begin{equation*}
K_{ijk} = \{(x,y,z) : i2^{-n}\ {\leqslant}\ x\ {\leqslant}\ (i+1)2^{-n},\ j2^{-n}\ {\leqslant}\ y\ {\leqslant}\ (j+1)2^{-n},\ k2^{-n}\ {\leqslant}\ z\ {\leqslant}\ (k+1)2^{-n}\},
\end{equation*}

where $i,j,k$ are integers in [$2^n$], the value of $\phi(c_{ijk})$ is $c_{ijk}+\delta_{ijk}$, where $\delta_{ijk}$ is one the following four vectors :
\begin{itemize}
\item{$\delta_0$ = ($-\alpha$,$-\alpha$,$-\alpha$)}
\item{$\delta_1$ = ($\alpha$,0,0)}
\item{$\delta_2$ = (0,$\alpha$,0)}
\item{$\delta_3$ = (0,0,$\alpha$)}
\end{itemize}

Here $\alpha > 0$ is much smaller than the cubelet side, say $2^{-2n}$. To compute $\phi$ at the centers of the cubelet $K_{ijk}$ we only need to know which of the four displacements to add. This is computed by a circuit $\mathcal{C}$ with $3n$ input bits and 2 output bits. $\mathcal{C}(i,j,k)$ is the index $r$ such that, if $c$ is the center of cubelet $K_{ijk}$, $\phi(c) = c + \delta_r$. $\mathcal{C}$ is such that $\mathcal{C}(0,j,k)=1$, $\mathcal{C}(i,0,k)=2$, $\mathcal{C}(i,j,0)=3$ (with conflicts resolved arbitrarily) and $\mathcal{C}(2^n-1,j,k) = \mathcal{C}(i,2^n-1,k) = \mathcal{C}(i,j,2^n-1)=0$, so that the function $\phi$ maps the boundary to the interior of the cube. A vertex of a cubelet is called {\em{panchromatic}} if among the eight cubelets adjacent to it there are four that have all four increments $\delta_0,\delta_1,\delta_2,\delta_3$. $\mathcal{C}$ is the only input to {\sc{3-dimensional Brouwer}}.

\begin{framed}
\noindent {\sc{3-Dimensional Brouwer}} : Given a circuit $\mathcal{C}$ as described above, find a panchromatic vertex.
\end{framed}

Daskalakis, Goldberg and Papadimitriou \cite{nash-ppad} proved that {\sc{3-dimensional Brouwer}} is PPAD-complete. They reduced {\sc{3-dimensional Brouwer}} to {\sc{3-Graphical Nash}} to prove that {\sc{3-Graphical Nash}} is PPAD-complete. We follow their approach to reduce {\sc{3-dimensional Brouwer}} to {\sc{3-Kernel Nash}}.

\begin{theorem}\label{thm:main}
{\sc{3-Dimensional Brouwer}} $\leq_P$ {\sc{3-Kernel Nash}}
\end{theorem}
\begin{proof}
\begin{figure}[hc]
  \centering
  \includegraphics[width=6in]{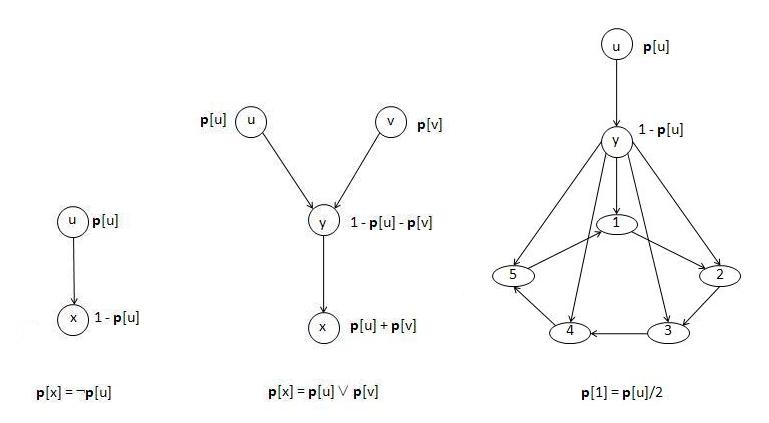}
  \caption{$\mathcal{G}_\neg$, $\mathcal{G}_\vee$, $\mathcal{G}_{\frac{1}{2}}$}%
  \label{fig:gadgets}%
\end{figure}
We reduce {\sc{3-dimensional Brouwer}} to {\sc{3-Kernel Nash}}. Given a circuit $\mathcal{C}$ with $3n$ input bits describing a Brouwer function, we shall construct an instance of {\sc{3-Kernel Nash}} that simulates $\mathcal{C}$. We construct instances of {\sc{3-Kernel Nash}} to simulate the required arithmetic ($=, -, +, <$, multiplication and division by 2) and boolean operations ($\vee, \wedge, \neg$). We then use the framework of \cite{nash-ppad}, \cite{normal-graphic} to combine these gadgets and simulate $\mathcal{C}$ and encode the geometric condition of fixed points in {\sc{3-dimensional Brouwer}}. Let $f(u)$ be the cost of a vertex $u$ in any Nash equilibrium. Henceforth we denote $f(u)$ by {\bf{p}}$[u]$\footnote{The reason behind this notation is that {\bf{p}}$[u]$ corresponds to the ``payoffs" in the graphical games of \cite{nash-ppad}. This makes our proof easier to understand using the simulation presented in \cite{nash-ppad}}. Let $\mathcal{G}_=, \mathcal{G}_-, \mathcal{G}_+, \mathcal{G}_<, \mathcal{G}_2, \mathcal{G}_{\frac{1}{2}}, \mathcal{G}_\vee, \mathcal{G}_\wedge, \mathcal{G}_\neg$ represent the required gadgets. All our gadgets are clique-acyclic, every maximal clique is of size at most 3, and all proper cycles are homogeneous and node-disjoint. Also, the costs of the input players do not depend on the costs of the output players. \\

\noindent {\bf{Boolean Operators}} : In the boolean gadgets, {\bf{p}}$[u]$ and {\bf{p}}$[v]$ are in $\{0,1\}$. Figure \ref{fig:gadgets} shows the construction of boolean gadgets $\mathcal{G}_\neg$ and $\mathcal{G}_\vee$. There are two input nodes $u$ and $v$ for $\mathcal{G}_\vee$ and only input node for $\mathcal{G}_\neg$. The value of {\bf{p}}$[x]$ is the result of applying the corresponding boolean function to the inputs {\bf{p}}$[u]$ and {\bf{p}}$[v]$. $\mathcal{G}_\wedge$ can be simulated using $\mathcal{G}_\neg$ and $\mathcal{G}_\vee$. \\

\noindent {\bf{Arithmetic Operators}} : In the arithmetic gadgets, {\bf{p}}$[u]$ and {\bf{p}}$[v]$ are real numbers in [0,1]. Note that $\mathcal{G}_=$ can be simulated using two $\mathcal{G}_\neg$ gadgets. $\mathcal{G}_+$ is same as $\mathcal{G}_\vee$. $\mathcal{G}_-$ can be simulated using two $\mathcal{G}_\neg$ gadgets and one $\mathcal{G}_+$. $\mathcal{G}_2$ can be simulated using $\mathcal{G}_=$ and $\mathcal{G}_+$. Figure \ref{fig:gadgets} shows our construction of $\mathcal{G}_{\frac{1}{2}}$. In this gadget, it is easy to see that in any Nash equilibrium {\bf{p}}$[1]$ = {\bf{p}}$[2]$ = {\bf{p}}$[3]$ = {\bf{p}}$[4]$ = {\bf{p}}$[5]$. Also {\bf{p}}$[u]$ = {\bf{p}}$[1]$ + {\bf{p}}$[2]$ = {\bf{p}}$[2]$ + {\bf{p}}$[3]$ = {\bf{p}}$[3]$ + {\bf{p}}$[4]$ = {\bf{p}}$[4]$ + {\bf{p}}$[5]$ = {\bf{p}}$[5]$ + {\bf{p}}$[1]$. Hence, {\bf{p}}$[1]$ = {\bf{p}}$[2]$ = {\bf{p}}$[3]$ = {\bf{p}}$[4]$ = {\bf{p}}$[5]$ = {\bf{p}}$[u]$/2. To simulate $\mathcal{G}_{<}$ using our arithmetic gadgets we create a node with ${\bf{p}}[x]=2^{-{\log}{\epsilon}}({\bf{p}}[v]-{\bf{p}}[u])$ for any given $\epsilon$. This is done using one $\mathcal{G}_-$ gadget and ${\log}{\epsilon}$ $\mathcal{G}_{\frac{1}{2}}$ gadgets. We also need a player $u$ with {\bf{p}}$[u]$ = $\frac{1}{2}$. Any node (say $v$) with zero indegree must have {\bf{p}}$[v]$ = 1 in any Nash equilibrium. Hence we can construct a node $u$ with {\bf{p}}$[u]$ = $\frac{1}{2}$ using the gadget $\mathcal{G}_{\frac{1}{2}}$.

To reduce {\sc{3-Dimensional Brouwer}} to {\sc{3-Kernel Nash}} we require three players representing the three coordinates. If a function $F$ is a Nash equilibrium then {\bf{p}}$[u]$ of each coordinate player is equal to its coordinate of the solution vertex of {\sc{3-Dimensional Brouwer}}. The rest of the proof is exactly same as in \cite{nash-ppad}. We refer the reader to Section 4 of \cite{nash-ppad} for the complete simulation and details of handling brittle comparators.
\end{proof}

\begin{corollary}
{\sc{3-Kernel Nash}}, {\sc{3-Strong Kernel}} and {\sc{Scarf}} are PPAD-complete.
\end{corollary}
\begin{proof}
We have {\sc{3-Dimensional Brouwer}} $\leq_P$ {\sc{3-Kernel Nash}} $\leq_P$ {\sc{3-Strong Kernel}} $\leq_P$ {\sc{Scarf}} $\leq_P$ {\sc{End Of The Line}}. We know that {\sc{3-Dimensional Brouwer}} is PPAD-complete \cite{nash-ppad}. Hence, {\sc{3-Kernel Nash}}, {\sc{3-Strong Kernel}} and {\sc{Scarf}} are PPAD-complete.
\end{proof}

\begin{corollary}
{\sc{Strong Kernel}} is PPAD-hard.
\end{corollary}
\begin{proof}
{\sc{3-Strong Kernel}} is a special case of {\sc{Strong Kernel}}. Hence, {\sc{Strong Kernel}} is PPAD-hard. \end{proof}

\section{Fractional Stable Paths Problem}\label{sec:fspp}

Stable Paths Problem (Griffin, Shepherd and Wilfong \cite{gsw}) and Fractional Stable Paths Problem (Haxell and Wilfong \cite{fbgp}) are defined in the context of interdomain routing. Let $G(V,E)$ be a simple graph, where $V$ is a set of $n$ source nodes and a unique destination node $d$, and $E$ is the set of edges. The source nodes attempt to establish a {\em{fractional}} paths to the destination node $d$. For any node $u$, $N(u)$ = \{$w\ |\ (u,w) \in E$\} is the set of neighbors of $u$. A $path$ from $s$ to $t$ is defined as a sequence of nodes $(v_1,v_2, \dots v_{k-1},v_k)$, where $v_1 = s$ and $v_k = t$ and $(v_i,v_{i+1}) \in E$ for $1 \leqslant i \leqslant k-1$. We assume that all paths are simple i.e., they do not have repeated nodes. An empty path has no edges and is denoted by $\phi$. Let $|P|$ denote the length of $P$, i.e., number of edges in $P$. If $P$ and $Q$ are non-empty paths such that the first node in $Q$ is same as the last node in $P$, then $PQ$ denotes the path formed by concatenating these paths. We say that path $R$ {\em{ends}} with path $Q$ if $R$ can be written as $PQ$. We say that $Q$ is a final segment of $R$. If $Q$ is non-empty, we say $Q$ is a {\em{proper}} final segment of $R$. \\

\noindent {\em{Preferred Paths}} : Each source node $v$ has a set of preferred paths denote by $\mathcal{P}^v$. Let $|\mathcal{P}^v|$ denote the number of permitted paths at node $v$. We assume that $\forall v \in V$, $\phi \in \mathcal{P}^v$ and we do not count $\phi$ in $|\mathcal{P}^v|$. If $P = (v,v_1,v_2,{\dots}v_k,d)$ is in $\mathcal{P}^v$, then the node $v_1$ is called the {\em{next-hop}} of path $P$. We assume that $\mathcal{P}^d = \emptyset$. \\

\noindent {\em{Ranking Function}} : Each node has an ordered list of its preferred paths. For $P,P'\in\mathcal{P}^v$, we denote $P' \leqslant_v P$  to mean that $v$ prefers $P'$ to $P$. For each $P \in \mathcal{P}^u$ let $\lambda^v(P)$ be $k$ if there are $k-1$ paths $P' \in \mathcal{P}^v$ such that $\lambda^v(P') < \lambda^v(P)$. We assume that ${\lambda}^v(\phi) = 0$. \\

\noindent {\em{FSPP Instance}} : Let $\mathcal{P} = \{\mathcal{P}^v\ |\ v \in V - \{d\}\}$. Let $\Lambda = \{{\lambda}^v\ |\ v \in V - \{d\}\}$. An instance of the Stable Paths Problem, $\mathcal{I} = {\langle}G, \mathcal{P}, {\Lambda}{\rangle}$, is a graph together with the permitted paths and the ranking functions at each node. \\

\noindent {\em{Feasible Solution}} : A feasible solution is defined as an assignment of a non-negative weight $w(P)$ to each path $P \in \mathcal{P}^v$, for every $v$ so that the weights satisfy the two properties listed below. For a non-empty path $S$, let $\displaystyle\mathcal{P}^{v}_{S}$ denote the set of paths in $\mathcal{P}^v$ that end with the path $S$.
\begin{itemize}
\item{{\em{Unity condition}} : For each node $v$, $\displaystyle\sum_{P \in \mathcal{P}^v} w(P) \leqslant 1$.}
\item{{\em{Tree condition}} : For each node $v$, and each non-empty path $S$, $\displaystyle\sum_{P \in \displaystyle\mathcal{P}^{v}_{S}}w(P)\leqslant w(S)$.}
\end{itemize}
\vspace{0.15in}
\noindent {\em{Fractional Stable Solution}} : A fractional stable solution is a feasible solution such that for any path $Q \in \mathcal{P}^v$, one of the two following conditions holds:

\begin{itemize}
\item{$\displaystyle\sum_{P \in \mathcal{P}^v} w(P) = 1$, and each $P \in \mathcal{P}^v$ with
$w(P) > 0$ is such that ${\lambda}^v(P) \geqslant {\lambda}^v(Q)$.
}
\item{there exists a proper final segment $S$ of $Q$, such that $\displaystyle\sum_{P \in \displaystyle\mathcal{P}^{v}_{S}}w(P) = w(S)$, and moreover each $P \in \displaystyle\mathcal{P}^{v}_{S}$ with $w(P) > 0$ is such that ${\lambda}^v(P) \geqslant {\lambda}^v(Q)$.
}
\end{itemize}

\begin{framed}
\noindent {\sc{Fspp}} : Given an instance $\mathcal{I} = {\langle}G, \mathcal{P}, {\Lambda}{\rangle}$ of FSPP, find a fractional stable solution.
\end{framed}

Haxell and Wilfong \cite{fbgp} proved that every instance of FSPP has a fractional stable solution. Their proof works in two stages. In the first stage, they use Scarf's lemma to show that every instance of FSPP has an $\epsilon$-solution, for any positive constant $\epsilon$. Then they apply a standard compactness-type argument to conclude that every instance has an exact solution. 
The following theorem is based on a reduction of Haxell and Wilfong \cite{fbgp-journal}.

\begin{theorem}
{\sc{Fspp}} is PPAD-hard.
\end{theorem}
\begin{proof}
We show that {\sc{3-Kernel Nash}} $\leq_P$ {\sc{Fspp}}. Let $D=(V,A)$ be a digraph. We construct an instance of FSPP $\mathcal{I} = {\langle}G, \mathcal{P}, {\Lambda}{\rangle}$. Let $G = (V \cup \{d\},E)$. We add an edge $(u,v)$ in $E$ if there is an arc $(u,v)$ or $(v,u)$ in $A$. For each $v \in V$ we add an edge $(v,d)$ in $E$. The preferred paths of $v$ consist of the path $vd$ and $vud$ such that $(v,u) \in A$. The path $vd$ is preferred the least. The preference among the rest of the paths (of the form $vud$) is arbitrary. Let $w$ be any fractional stable solution of $\mathcal{I}$. The corresponding solution to {\sc{3-Kernel Nash}} is obtained by setting $f(v) = w(vd)$ for each $v \in V$.
\end{proof}

\section{Related Problems in PPAD}\label{sec:related}

\begin{comment}
\subsection{Finding Core in a Balanced Game}

\begin{framed}
\noindent {\sc{Core}} : Given a balanced $N$-person game find its core.
\end{framed}

Scarf studied the core of an $N$-person game and proved that {\em{every balanced $N$-person game has a nonempty core}} \cite{scarf}. Scarf proved that {\sc{Core}} $\leq_P$ {\sc{Scarf}}. By Theorem \ref{thm:scarf}, we have {\sc{Scarf}} $\leq_P$ {\sc{End Of The Line}}. Hence {\sc{Core}} is in PPAD.
\end{comment}

\subsection{Hypergraphic Preference Systems}

A hypergraphic preference system is a pair $(H, \mathcal{O})$, where $H = (V,E)$ is a hypergraph, and $\mathcal{O} = \{\leq_v\ :\ v \in V\}$ is a family of linear orders, $\leq_v$ being an order on the set $D(v)$ of edges containing the vertex $v$. A set $M$ of edges is called a stable matching with respect to the preference system if it is a matching and for every edge $e$ there exists a vertex $v \in e$ and an edge $m \in M$ containing $v$ such that $e \leq_v m$. A nonnegative function $w$ on the edges in $H$ is called a fractional matching if $\sum_{v \in h}{w(h)} \leqslant 1$ for every vertex $v$. A fractional matching $w$ is called {\em{stable}} if every edge $e$ contains a vertex $v$ such that $\sum_{v \in h, e \leq_v h}{w(h)} = 1$.

\begin{framed}
\noindent {\sc{Hypergraphic Fractional Stable Matching}} : Given a hypergraphic preference system $(H, \mathcal{O})$ find a fractional stable matching.
\end{framed}

Aharoni and Fleiner \cite{onlemma} proved the following theorem. We refer the reader to \cite{onlemma} for the details of the proof. We observe that their proof implies {\sc{Hypergraphic Fractional Stable Matching}} $\leq_P$ {\sc{Scarf}}.

\begin{theorem} (Aharoni and Fleiner \cite{onlemma})
Every hypergraphic preference system has a fractional stable matching
\end{theorem}

\begin{corollary}
{\sc{Hypergraphic Fractional Stable Matching}} $\in$ PPAD.
\end{corollary}

\subsection{Approximate FSPP}

There are two notions of approximation for FSPP : $\epsilon$-solution \cite{fbgp} and $\epsilon$-stable solution \cite{kintali-fspp}. \\

\noindent {\bf{$\epsilon$-solution}} : An $\epsilon$-solution is defined as an assignment of a non-negative weight $w(P)$ to each path $P \in \mathcal{P}^v$, for every $v$ so that the weights satisfy the properties listed below. For a non-empty path $S$, let $\displaystyle\mathcal{P}^{v}_{S}$ denote the set of paths in $\mathcal{P}^v$ that end with the path $S$.
\begin{itemize}
\item{{\em{Unity condition}} : For each node $v$, $\displaystyle\sum_{P \in \mathcal{P}^v} w(P) \leqslant 1$.}
\item{{\em{$\epsilon$-Tree condition}} : For each node $v$, and each non-empty path $S$, $\displaystyle\sum_{P \in \displaystyle\mathcal{P}^{v}_{S}}w(P)\leqslant w(S)+\epsilon$.}

\item {{\em{Stability condition}} : For any path $Q \in \mathcal{P}^v$, one of the two following conditions holds:

\begin{itemize}
\item{$\displaystyle\sum_{P \in \mathcal{P}^v} w(P) = 1$, and each $P \in \mathcal{P}^v$ with
$w(P) > 0$ is such that ${\lambda}^v(P) \geqslant {\lambda}^v(Q)$.
}
\item{there exists a proper final segment $S$ of $Q$, such that $\displaystyle\sum_{P \in \displaystyle\mathcal{P}^{v}_{S}}w(P) = w(S)+\epsilon$, and moreover each $P \in \displaystyle\mathcal{P}^{v}_{S}$ with $w(P) > 0$ is such that ${\lambda}^v(P) \geqslant {\lambda}^v(Q)$.
}
\end{itemize}}
\end{itemize}
\vspace{0.15in}

\begin{framed}
\noindent {\sc{$\epsilon$-solution of FSPP}} : Given and instance of FSPP find an $\epsilon$-solution.
\end{framed}

Using Scarf's lemma, Haxell and Wilfong \cite{fbgp} proved that every instance of FSPP has an $\epsilon$-solution. Their proof is a polynomial reduction from {\sc{$\epsilon$-solution of FSPP}} to {\sc{Scarf}}. For more details we refer the reader to \cite{fbgp}.

\begin{theorem} (Haxell and Wilfong \cite{fbgp})
Every instance of FSPP has an $\epsilon$-solution.
\end{theorem}

\begin{corollary}
{\sc{$\epsilon$-solution of FSPP}} $\in$ PPAD.
\end{corollary}

\noindent {\bf{$\epsilon$-stable Solution}} :
An $\epsilon$-stable solution to FSPP is a feasible solution (i.e., it satisfies unity and tree conditions mentioned in Section \ref{sec:fspp}) such that for any path $Q \in \mathcal{P}^v$, one of the two following conditions holds:

\begin{itemize}
\item{$1 - \epsilon \leqslant \displaystyle\sum_{P \in \mathcal{P}^v} w(P) \leqslant 1$, and each $P \in \mathcal{P}^v$ with
$w(P) > 0$ is such that ${\lambda}^v(P) \geqslant {\lambda}^v(Q)$.
}
\item{there exists a proper final segment $S$ of $Q$, such that $w(S) - \epsilon \leqslant \displaystyle\sum_{P \in \displaystyle\mathcal{P}^{v}_{S}}w(P) \leqslant w(S)$, and moreover each $P \in \displaystyle\mathcal{P}^{v}_{S}$ with $w(P) > 0$ is such that ${\lambda}^v(P) \geqslant {\lambda}^v(Q)$.
}
\end{itemize}

Note that, when $\epsilon = 0$, both $\epsilon$-stable solution and $\epsilon$-solution are equivalent to a fractional stable solution. In \cite{kintali-fspp} we defined a game-theoretic model of FSPP and presented a relation between $\epsilon$-Nash and $\epsilon$-stable solution. We also presented a constructive proof (a distributed algorithm) showing that all instances of FSPP have an $\epsilon$-stable solution for any given $\epsilon > 0$. However, the complexity of finding an $\epsilon$-stable solution is an open problem.

\section{Conclusion and Open Problems}\label{sec:open}

In this paper, we studied the complexity of computational version of Scarf's lemma ({\sc{Scarf}}) and related problems. We proved that {\sc{Scarf}} is complete for the complexity class PPAD thus showing that {\sc{Scarf}} is as hard as the computational versions of Brouwer's fixed point theorem and Sperner's lemma. Hence, there is no polynomial-time algorithm for {\sc{Scarf}} unless PPAD $\subseteq$ P. We also showed that fractional stable paths problem and finding strong fractional kernels in digraphs are PPAD-hard. Following are some of the problems left open by our work :

\begin{itemize}
\item{We know that {\sc{Hypergraphic Fractional Stable Matching}} is in PPAD. Is it PPAD-complete ?}
\item{What is the complexity of finding $\epsilon$-solution \cite{fbgp} or $\epsilon$-stable solution \cite{kintali-fspp} of FSPP ?}
\item{Is {\sc{Fspp}} in PPAD ? A positive answer would give an alternate proof of Haxell and Wilfong's theorem.}
\item{What is the complexity of finding a core in a balanced $N$-person game with nontransferable utilities ?}
\end{itemize}

\noindent Please see \cite{kprst-ppad} for a complete treatment of the above open problems.

\vspace{0.30in}

\noindent {\large{\bf{Acknowledgements}}} : I am grateful to H. Venkateswaran for many helpful and motivating discussions throughout the course of this project. I would like to thank Gordon Wilfong for sending preprints of \cite{fbgp-journal} and \cite{fbgp}.

\bibliographystyle{plain}
\bibliography{bib-kintali}

\begin{thebibliography}{10}

\bibitem{onlemma}
Ron Aharoni and Tam$\acute{\mbox{a}}$s Fleiner.
\newblock {O}n a {L}emma of {S}carf.
\newblock {\em Journal of Combinatorial Theory, Series B}, 87(1):72–80, 2003.

\bibitem{frackernel}
Ron Aharoni and Ron Holzman.
\newblock {F}ractional {K}ernels in {D}igraphs.
\newblock {\em Journal of Combinatorial Theory, Series B}, 73(1):1--6, 1998.

\bibitem{3nash-chen}
Xi~Chen and Xiaotie Deng.
\newblock 3-{NASH} is {PPAD}-complete.
\newblock {\em Electronic Colloquium on Computational Complexity (ECCC)}, 134,
  2005.

\bibitem{2nash-chen}
Xi~Chen and Xiaotie Deng.
\newblock Settling the complexity of two-player nash equilibrium.
\newblock {\em In Proceedings of FOCS}, pages 261--272, 2006.

\bibitem{apxnash-chen}
Xi~Chen, Xiaotie Deng, and Shang-Hua Teng.
\newblock Computing nash equilibria: Approximation and smoothed complexity.
\newblock {\em FOCS}, pages 603--612, 2006.

\bibitem{nash-ppad}
Constantinos Daskalakis, Paul~W. Goldberg, and Christos~H. Papadimitriou.
\newblock The {C}omplexity of {C}omputing a {N}ash {E}quilibrium.
\newblock {\em In Proceedings of the thirty-eighth annual ACM Symposium on
  Theory of Computing}, pages 71--78, 2006.

\bibitem{3nash-papa}
Konstantinos Daskalakis and Christos~H. Papadimitriou.
\newblock Three-player games are hard.
\newblock {\em Electronic Colloquium on Computational Complexity (ECCC)}, 139,
  2005.

\bibitem{gale-shapley}
D.~Gale and L.~S. Shapley.
\newblock College admissions and the stability of marriage.
\newblock {\em American Mathematical Monthly}, 69:9--14, 1962.

\bibitem{normal-graphic}
Paul~W. Goldberg and Christos~H. Papadimitriou.
\newblock {R}educibility among {E}quilibrium {P}roblems.
\newblock {\em STOC 2006}, pages 61--70.

\bibitem{gsw}
T.~Griffin, F.~B. Shepherd, and G.~Wilfong.
\newblock The {S}table {P}aths {P}roblem and {I}nterdomain {R}outing.
\newblock {\em IEEE/ACM Transactions on Networking}, 10(2):232--243, 2002.

\bibitem{fbgp-journal}
Penny~E. Haxell and Gordon~T. Wilfong.
\newblock {O}n the {S}table {P}aths {P}roblems.
\newblock {\em (To Appear). Preliminary version appeared in SODA 2008}.

\bibitem{fbgp}
Penny~E. Haxell and Gordon~T. Wilfong.
\newblock {A} {F}ractional {M}odel of the {B}order {G}ateway {P}rotocol.
\newblock {\em SODA}, 2008.

\bibitem{class-pls}
David~S. Johnson, Christos~H. Papadimitriou, and Mihalis Yannakakis.
\newblock {H}ow {E}asy is {L}ocal {S}earch?
\newblock {\em Journal of Computer and System Sciences}, 37(1):79--100, 1988.

\bibitem{kintali-fspp}
Shiva Kintali.
\newblock {A} {D}istributed {P}rotocol for {F}ractional {S}table {P}aths
  {P}roblem.
\newblock {\em Presented at the DIMACS/DyDAn Workshop on Secure Internet
  Routing, Rutgers University}, March 24-26 2008.

\bibitem{kprst-ppad}
Shiva Kintali, Laura~J. Poplawski, Rajmohan Rajaraman, Ravi Sundaram, and
  Shang-Hua Teng.
\newblock Reducibility among fractional stability problems.
\newblock {\em (CoRR) abs/0904.1435, (Submitted April 2009)}, 2009.

\bibitem{h-perfect}
Tam$\acute{\mbox{a}}$s Kir$\acute{\mbox{a}}$ly and J$\acute{\mbox{u}}$lia Pap.
\newblock {K}ernels, {S}table {M}atchings, and {S}carf's {L}emma.
\newblock {\em The Egerváry Research Group Technical Report TR-2008-13,
  http://www.cs.elte.hu/egres/tr/egres-08-13.pdf}.

\bibitem{class-tfnp}
Nimrod Megiddo and Christos~H. Papadimitriou.
\newblock {O}n {T}otal {F}unctions, {E}xistence {T}heorems and {C}omputational
  {C}omplexity.
\newblock {\em Theoretical Computer Science}, 81(2):317--324, 1991.

\bibitem{agt-book}
N.~Nisan, T.~Roughgarden, E.~Tardos, and V.~Vazirani.
\newblock Algorithmic game theory.
\newblock {\em Cambridge University Press}, 2007.

\bibitem{class-ppad}
Christos~H. Papadimitriou.
\newblock On the {C}omplexity of the {P}arity {A}rgument and {O}ther
  {I}nefficient {P}roofs of {E}xistence.
\newblock {\em J. Comput. Syst. Sci.}, 48(3):498--532, 1994.

\bibitem{scarf}
Herbert~E. Scarf.
\newblock The {C}ore of an {N} {P}erson {G}ame.
\newblock {\em Econometrica}, 69:35:50, 1967.

\bibitem{schrijver-book}
Alexander Schrijver.
\newblock {C}ombinatorial {O}ptimization, {P}olyhdera and {E}fficiency,
  {V}olume {B}.
\newblock {\em Springer-Verlag Berlin Heidelberg}, 2003.

\bibitem{ShapleyOrientation}
L.~S. Shapley.
\newblock A note on the lemke-howson algorithm.
\newblock {\em Mathematical Programming Study}, 1:175--189, 1974.

\bibitem{ToddOrientation}
Michael~J. Todd.
\newblock Orientation in {C}omplementary {P}ivot {A}lgorithms.
\newblock {\em Mathematics of Operations Research}, 1(1):54--66, 1976.

\end{thebibliography}

\vspace{0.30in}

%\noindent {\bf{\LARGE{Appendix}}} \\

\end{document}